\makeatletter
\def\input@path{{styles/}}
\makeatother

\ifx\SODA\undefined
    \newcommand{\SODAVer}[1]{}
    \newcommand{\NotSODAVer}[1]{#1}
\else
    
    \newcommand{\SODAVer}[1]{#1}
    \newcommand{\NotSODAVer}[1]{}
\fi

\ifx\GenSoCGVer\undefined%
\documentclass[11pt]{article}%
\newcommand{\SoCG}[1]{}
\newcommand{\NotSoCG}[1]{#1}%
\else
\newcommand{\SoCG}[1]{#1}
\newcommand{\NotSoCG}[1]{}%
\fi

\usepackage{microtype}
\usepackage{placeins}

\SoCG{%
  \documentclass[a4paper,thm-restate,anonymous]{socg-lipics-v2021}%
   \usepackage{fontawesome5}
}%

\NotSoCG{%
   \def\UseBibLatex{1}%
}

\providecommand{\BibLatexMode}[1]{}
\providecommand{\BibTexMode}[1]{}

\renewcommand{\BibLatexMode}[1]{#1}
\renewcommand{\BibTexMode}[1]{}

\ifx\UseBibLatex\undefined%
  \renewcommand{\BibLatexMode}[1]{}
  \renewcommand{\BibTexMode}[1]{#1}
\fi

\BibLatexMode{%
   \usepackage[bibencoding=utf8,style=alphabetic,backend=biber]{biblatex}%
   \usepackage{sariel_biblatex}%
}

\SoCG{%

   \pgfkeys{/prAtEnd/global custom defaults/.style={
         text link=,end,restate
      }
   }%
}

\NotSoCG{%
}
\usepackage{amsmath}%
\usepackage{amssymb}%
\usepackage[table]{xcolor}%

\NotSoCG{%
   \usepackage[amsmath,thmmarks]{ntheorem}%
}

\usepackage{todonotes}%
\NotSoCG{\usepackage{titlesec}}%
\usepackage{xcolor}%
\usepackage{mleftright}%
\usepackage{xspace}%
\usepackage{graphicx}
\usepackage{hyperref}%
\usepackage[inline]{enumitem}
\usepackage{hyperref}%
\usepackage[ocgcolorlinks]{ocgx2}
\usepackage{subcaption}

\hypersetup{%
      unicode,
      breaklinks,%
      colorlinks=true,%
      urlcolor=[rgb]{0.25,0.0,0.0},%
      linkcolor=[rgb]{0.5,0.0,0.0},%
      citecolor=[rgb]{0,0.2,0.445},%
      filecolor=[rgb]{0,0,0.4},
      anchorcolor=[rgb]={0.0,0.1,0.2}%
}

\NotSoCG{%
\titlelabel{\thetitle. }%

\theoremseparator{.}%

\theoremstyle{plain}%
\newtheorem{theorem}{Theorem}[section]

\newtheorem{lemma}[theorem]{Lemma}

\newtheorem{corollary}[theorem]{Corollary}

\newtheorem{observation}[theorem]{Observation}

\theoremstyle{plain}%
\theoremheaderfont{\sf} \theorembodyfont{\upshape}%
\newtheorem*{remark:unnumbered}[theorem]{Remark}%
\newtheorem{remark}[theorem]{Remark}%
\newtheorem{definition}[theorem]{Definition}

\newtheorem{example}[theorem]{Example}

\newtheorem{problem}[theorem]{Problem}

}

\NotSoCG{%
\theoremheaderfont{\em}%
\theorembodyfont{\upshape}%
\theoremstyle{nonumberplain}%
\theoremseparator{}%
\theoremsymbol{\myqedsymbol}%
\newtheorem{proof}{Proof:}%
}

\SoCG{%
   \newtheorem{problem}[theorem]{Problem}
}

\providecommand{\emphind}[1]{}%
\renewcommand{\emphind}[1]{\emph{#1}\index{#1}}

\definecolor{blue25emph}{rgb}{0, 0, 11}

\providecommand{\emphic}[2]{}
\renewcommand{\emphic}[2]{\textcolor{blue25emph}{%
      \textbf{\emph{#1}}}\index{#2}}

\providecommand{\emphi}[1]{}%
\renewcommand{\emphi}[1]{\emphic{#1}{#1}}

\definecolor{almostblack}{rgb}{0, 0, 0.3}

\providecommand{\emphw}[1]{}%
\renewcommand{\emphw}[1]{{\textcolor{almostblack}{\emph{#1}}}}%

\providecommand{\emphOnly}[1]{}%
\renewcommand{\emphOnly}[1]{\emph{\textcolor{blue25}{\textbf{#1}}}}

\newcommand{\myqedsymbol}{\rule{2mm}{2mm}}
\newcommand{\SamThanks}[1]{%
   \thanks{%
      Siebel School of Computing and Data Science; %
      University of Illinois; %
      201 N. Goodwin Avenue; %
      Urbana, IL, 61801, USA; %
      \href{mailto:spam@illinois.edu}{samuelr6@illinois.edu}; %
      \url{https://surg.dev/}. %
   #1%
   }%
}
\newcommand{\SarielThanks}[1]{%
   \thanks{%
      Department of Computer Science; %
      University of Illinois; %
      201 N. Goodwin Avenue; %
      Urbana, IL, 61801, USA; %
      \href{mailto:spam@illinois.edu}{sariel@illinois.edu}; %
      \url{http://sarielhp.org/}.%
   #1%
   }%
}

\newcommand{\HLink}[2]{\hyperref[#2]{#1~\ref*{#2}}}
\newcommand{\HLinkY}[2]{\hyperref[#2]{#1}}
\newcommand{\HLinkSuffix}[3]{\hyperref[#2]{#1\ref*{#2}{#3}}}

\newcommand{\figlab}[1]{\label{fig:#1}}
\newcommand{\figref}[1]{\HLink{Figure}{fig:#1}}

\newcommand{\thmlab}[1]{{\label{theo:#1}}}
\newcommand{\thmref}[1]{\HLink{Theorem}{theo:#1}}

\newcommand{\seclab}[1]{\label{sec:#1}}
\newcommand{\secref}[1]{\HLink{Section}{sec:#1}}

\newcommand{\tablab}[1]{\label{table:#1}}%
\newcommand{\tabref}[1]{\HLink{Table}{table:#1}}%

\providecommand{\deflab}[1]{\label{def:#1}}

\providecommand{\exalab}[1]{\label{example:#1}}
\newcommand{\exaref}[1]{\HLink{Example}{example:#1}}

\newcommand{\lemlab}[1]{\label{lem:#1}}
\newcommand{\lemref}[1]{\HLink{Lemma}{lem:#1}}%

\providecommand{\eqlab}[1]{}%
\renewcommand{\eqlab}[1]{\label{equation:#1}}
\newcommand{\Eqref}[1]{\HLinkSuffix{Eq.~(}{equation:#1}{)}}

\providecommand{\remove}[1]{}%
\newcommand{\Set}[2]{\left\{ #1 \;\middle\vert\; #2 \right\}}

\newcommand{\pth}[1]{\mleft(#1\mright)}%

\newcommand{\cardin}[1]{\left\lvert {#1} \right\rvert}%

\renewcommand{\th}{th\xspace}

\renewcommand{\Re}{\mathbb{R}}%

\newlist{compactenumA}{enumerate}{5}%
\setlist[compactenumA]{topsep=0pt,itemsep=-1ex,partopsep=1ex,parsep=1ex,%
   label=(\Alph*)}%

\newlist{compactenuma}{enumerate}{5}%
\setlist[compactenuma]{topsep=0pt,itemsep=-1ex,partopsep=1ex,parsep=1ex,%
   label=(\alph*)}%

\newlist{compactenumI}{enumerate}{5}%
\setlist[compactenumI]{topsep=0pt,itemsep=-1ex,partopsep=1ex,parsep=1ex,%
   label=(\Roman*)}%

\newlist{compactenumi}{enumerate}{5}%
\setlist[compactenumi]{topsep=0pt,itemsep=-1ex,partopsep=1ex,parsep=1ex,%
   label=(\roman*)}%

\newlist{compactitem}{itemize}{5}%
\setlist[compactitem]{topsep=0pt,itemsep=-1ex,partopsep=1ex,parsep=1ex,%
   label=\ensuremath{\bullet}}%

\NotSoCG{%
   \numberwithin{figure}{section}%
   \numberwithin{table}{section}%
   \numberwithin{equation}{section}%
}

\DeclareFontFamily{U}{BOONDOX-calo}{\skewchar\font=45 }
\DeclareFontShape{U}{BOONDOX-calo}{m}{n}{
  <-> s*[1.05] BOONDOX-r-calo}{}
\DeclareFontShape{U}{BOONDOX-calo}{b}{n}{
  <-> s*[1.05] BOONDOX-b-calo}{}
\DeclareMathAlphabet{\mathcalb}{U}{BOONDOX-calo}{m}{n}
\SetMathAlphabet{\mathcalb}{bold}{U}{BOONDOX-calo}{b}{n}
\DeclareMathAlphabet{\mathbcalb}{U}{BOONDOX-calo}{b}{n}

\newcommand{\model}{\ensuremath{\mathcal{M}}\xspace}%
\newcommand{\areaX}[1]{\mathsf{m}\pth{#1}}%
\newcommand{\bd}{\partial}
\newcommand{\lenX}[1]{\left\| #1 \right\|}
\newcommand{\dMC}{\mathsf{d}_{\model}}%
\newcommand{\dMY}[2]{\dMC\pth{#1,#2}}%
\newcommand{\Expansion}{\tau^\star}
\newcommand{\interX}[1]{\mathrm{int}\pth{#1}}%
\newcommand{\salientY}[2]{\mathsf{s}(#1, #2)}%

\newcommand{\neck}{\mathcalb{o}}%

\renewcommand{\H}{\mathsf{H}}

\newcommand{\disk}{\mathcalb{d}}%

\newcommand{\etal}{\textit{et~al.}\xspace}
\newcommand{\Hetroy}{H{\'{e}}troy\xspace}%

\DefineNamedColor{named}{OliveGreen} {cmyk}{0.64,0,0.95,0.40}
\providecommand{\ComplexityClass}[1]{{{\textcolor[named]{OliveGreen}{%
            \textsc{#1}}}}}
\providecommand{\NPHard}{{\ComplexityClass{NP-Hard}}%
   \index{NP!hard}\xspace}
\providecommand{\NPComplete}{\ComplexityClass{NP-Complete}%
   \index{NP!complete}\xspace}

\newcommand{\Term}[1]{\textsf{#1}}
\newcommand{\BFS}{\Term{BFS}\xspace}%
\newcommand{\MST}{\Term{MST}\xspace}%

\newcommand{\regX}[1]{\patch_{#1}}
\newcommand{\oregX}[1]{\opatch_{#1}}

\newcommand{\crA}{\xi}
\newcommand{\crB}{\varsigma}

\newcommand{\crD}{\eta}

\newcommand{\Path}{\kappa}

\newcommand{\GVE}{G=(V,E)}
\newcommand{\GLasso}{\varsigma}
\newcommand{\spathX}[1]{{\circlearrowright_{ \Path}}(#1)}
\newcommand{\neckC}{\mathcal{N}}%
\newcommand{\neckY}[2]{\mathcal{N}\pth{#1,#2}}%

\usepackage{stmaryrd}
\newcommand{\tightnessX}[1]{\left\langle {#1} \right\rangle}

\newcommand{\eps}{\varepsilon}%
\newcommand{\RP}{\mathcal{R}}

\newcommand{\patch}{\mathcalb{b}}%
\newcommand{\opatch}{\overline{\mathcalb{b}}}%
\newcommand{\patchM}{\patch_i}%
\newcommand{\collar}{collar\xspace}%

\newcommand{\optCollar}{\vartheta^\star}

\SoCG{%
   \newcommand{\mypar}[1]{%
      \subparagraph*{#1}%
   }
}

\NotSoCG{%
   \newcommand{\mypar}[1]{%
      \paragraph*{#1}%
   }
}

\BibLatexMode{%
   \bibliography{neck_cuts}%
}

\begin{document}

\title{In the Search for Good Neck Cuts}

\SoCG{%
   \ccsdesc{Computing methodologies~Shape modeling}
   \keywords{Constrictions, Object Representations, Computer Graphics}
   \Copyright{Sariel Har-Peled Sam Ruggerio}
   \authorrunning{Har-Peled, Ruggerio}%
}

\NotSoCG{%
   \author{%
      Sariel Har-Peled%
      \SarielThanks%
      {Work on this paper was partially supported by NSF AF award CCF-2317241.%
      }%
      \and%
      Sam Ruggerio%
      \SamThanks%
      {Work on this paper was partially supported by NSF AF award CCF-2317241.%
      }%
   }%
}%

\maketitle

\begin{figure}[h!]
    \centerline{   \includegraphics[width=\textwidth]{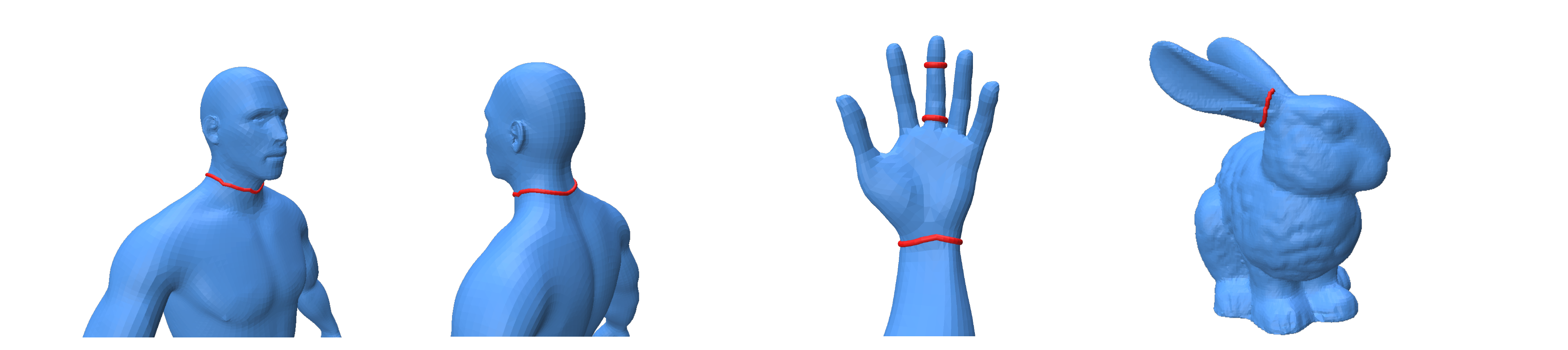}}
\end{figure}

\begin{abstract}
    We study the problem of finding neck-like features on a surface.  Applications for such cuts include robotics, mesh segmentation, and algorithmic applications. We provide a new definition for a surface bottleneck --- informally, it is the shortest cycle relative to the size of the areas it separates. Inspired by the isoperimetric inequality, we formally define such optimal cuts, study their properties, and present several algorithms inspired by these ideas that work surprisingly well in practice. For examples of our algorithms, see \url{https://neckcut.space}.
\end{abstract}

\section{Introduction}
\seclab{int}

Computing ``good'' cycles on surfaces is a well-studied problem \cite{eh-ocsid-04, cm-fsnsn-07, bgsf-rgsaa-08, cdem-fotc-10, dfw-echtl-13, ft-cls-13, cvl-fsntc-16, blr-fmsrg-25}, such as computing a class of cycles, such as shortest geodesic cycles, non-contractable loops, handles, etc. We are interested in cycles that represent neck-like features on a surface. Identifying neck-like features on a 3D surface mesh has been a crucial algorithmic problem in applications such as robotics, mesh segmentation, and more.  Neck-like cycles are often employed in intermediate steps within these applications, but computing them can be both challenging and time-consuming, as seen in \cite{zmt-fsptm-05, lf-isutp-07}. Existing methods tend to rely on expensive preprocessing for topological methods, which may also involve mesh modification, to produce neck-like cycles.

In this work, we consider two problems:
\begin{problem}
    What is the optimal notion of a neck-like surface, and the cycles to define these necks?
\end{problem}
\begin{problem}
    How can neck-like cycles be efficiently computed?
\end{problem}

We propose a new geometrically motivated definition of a \emphw{bottleneck curve} (or \textit{neck-cut}), based on the isoperimetric quantity. We describe a theoretical approximation algorithm to find near-optimal bottleneck curves, which runs in polynomial time. We then implemented a practical algorithm, with this motivating background, to run on real models, which runs in sub-quadratic time with good results. Our practical algorithm is simple to implement, relying only on shortest path algorithms and filtering to achieve the results shown, with no second-pass optimization or curve smoothing required.

\subsection{Background and prior work}

\mypar{Necks versus non-contractable loops.}

Finding neck-like features differs from finding the shortest non-contractable loops on a surface.  As a reminder, a non-contractable loop on a surface corresponds to a cycle that can not be morphed into a point. Naturally, a neck-like loop might lie on an object that is topologically a ball (as are most of the examples shown in figures throughout this paper) -- for example, on an hourglass, all the loops are contractable, yet it has a neck-like feature.  While in many high-genus objects, these non-contractable loops may act as neck-like curves, there may be other non-contractable loops that do not lie on a feature boundary, or contractable loops that are on neck-like features, which would not be considered. Finding non-contractible cycles can be used to reduce the genus of a surface, by cutting along them as shown in \cite{eh-ocsid-04}.

\mypar{Sparsest cut (Cheeger's constant).}

A natural approach to addressing this problem is to consider the model as a graph $\GVE$ and compute the sparsest cut.  That is, find a partition of the vertices of $G$ into two sets $S, \overline{S}$, such that the ratio
\begin{equation*}
    \phi(S,\overline{S})
    =
    \frac{\cardin{E(S,\overline{S})}}{\min( |S|, |\overline{S}| )}
\end{equation*}
is minimized. Sadly, the associated optimization problem is \NPHard, and instead one can use algebraic techniques to approximate it.  People use various heuristics inspired by this observation to find good cuts.  For example, Gotsman \cite{got-gp-03} noted the connection of the Cheeger constant to the Laplacian of a surface mesh. Using this, he was able to detect whether a small cut exists and how to partition the graph based on spectral embedding for genus-0 models.

However, the above definition of Cheeger's constant does not constrain that the cut is \textit{connected}. Additionally, transitioning to algebraic methods often results in fuzzy boundaries that require further refinement. Gotsman's work approximates the optimal Laplacian basis, as computing the optimal would be on the order of $O(|V|^2)$ time.

These techniques appear to yield relatively slow algorithms. Since they do not work directly with the geometry, the generated cuts, while of relatively good quality, are not quite locally optimal.

\mypar{Topological Methods.}

Abdelrahman and Tong \cite{at-fcnf-23} computed neck-like features on meshes by locating critical points in a volumetric mesh and generating cutting planes over the mesh to isolate these loops.  The main tool is identifying points that are $2$-saddles of a Morse function (generated by a distance function), as they are good seed locations for neck-like features. This work extended Feng and Tong \cite{ft-cls-13}, who evaluated the persistent homology of the mesh to locate neck-like loops.

Abdelrahman and Tong \cite{at-fcnf-23} achieve a speedup, compared to \cite{ft-cls-13}, by producing an initial neck loop from the cutting plane, which is smoothed out via shortest loop evaluation. The loops they generate are of good quality on all genus meshes. The main pitfall, however, is the need for a volumetric tetrahedral mesh, resulting in a substantial increase in vertex and face density. Other topological approaches, such as the one by Dey \etal \cite{dlsc-cghtl-08}, often have the same requirement of a volumetric mesh. Such a volumetric mesh can be substantially larger than the input 2d model, and computing them from an input surface is quite challenging in some cases.

\mypar{Surface Methods.}

\Hetroy and Attali \cite{h-c-05, ha-dccps-03, ha-cpgcc-03} compute geodesics on the surface, and slide to fit them to generate tight constrictions (neck-like features). Earlier works relied on mesh simplification to generate seed curves. However, in all cases, these algorithms rely on the local properties of geodesics to find neck loops.

Specifically, \Hetroy \cite{h-c-05} approximates the mean curvature of the mesh in all locations to find seed locations for constrictions. Then, the algorithm performs a local search from these seed locations until the constrictions are minimized, and smooths and minimizes the curve.  Xin \etal \cite{xhf-ecegl-12} designed a fast algorithm to find shortest, exact geodesics on a model, regardless of the quality of the input mesh. However, initial cutting loops must be specified in the input, as this algorithm does not discover loops from the input model directly.

Tierny \etal \cite{tvd-tdmse-06} use similar methods to our approach. However, when approximating constrictions, they use the concavity of the curve. This algorithm requires computation of the Discrete Gaussian Curvature \cite{mbsb-ddgo-03} of each point on a curve, requiring an $O(n^2)$ time algorithm to compute that metric.

\subsection{Our approach}

Our starting point is defining formally what constitutes a good neck cut. Intuitively, it is a curve that bounds a large area, while being short. In the plane, the largest area one can capture if the length of the perimeter is fixed is a disk. This innocent-looking observation is a consequence of the famous isoperimetric inequality.  It states that for any region $R$ in the plane, and any disk $\disk$ of radius $r$, we have that
\begin{equation*}
    \frac{\areaX{R}}{\lenX{\partial{R}}^2}
    \leq%
    \frac{\areaX{\disk}}{\lenX{\partial{\disk}}^2}
    =%
    \frac{\pi r^2}{(2\pi r)^2}
    =%
    \frac{1}{4\pi},
\end{equation*}
where $\areaX{R}$ denotes the \emphw{area} of $R$, and $\lenX{\crD}$ denotes the \emphw{length} of a curve $\crD$.  We view the ratio on the left as the \emphw{isoperimetric ratio} of the boundary of $R$. A good neck-cut would have a high ratio. A curve on a surface might bound a much larger area compared to its perimeter, but unlike the plane, we have to consider both sides bounded by the curve. As such, the \emphw{tightness} of a closed curve on a surface (of genus zero) is the minimum, among the two regions it bounds, of the isoperimetric ratio.

\mypar{Computing tightness.} %
Unfortunately, computing (or even approximating) the tightest cycle on a surface appears to be hopeless in terms of efficient algorithms.  Nevertheless, it provides us with an easily computable scoring function to compare cycles (i.e., the tighter, the better). There are cases where the optimal neck-cut is intuitively obvious, see \figref{ostrich}. We thus investigate sufficient conditions under which we can efficiently approximate the optimal neck-cut. To this end, we first formally define tightness in \secref{iso_necks}.

Specifically, for a loop, we look at the ratio between the area it encloses and its length squared (for a circle in the plane, this ratio is $1/4\pi$, as shown above). Clearly, the bigger the ratio, the better neck-like the cycle is. We formally define the underlying optimization problem in \secref{iso_necks}.

\mypar{Well-behaved surfaces, salient points, and discovering necks.}  We quantify, in \secref{well_behaved}, what it means for a surface to be well-behaved -- intuitively, it should have bounded growth (which all real-world surfaces seem to possess).  To discover the neck-cuts, we try to identify necks -- to this end, we study in \secref{salient} \emphw{salient} points that can be used to define necks. Intuitively, salient points are extremal points of the model (such as the tips of fingers in a human model). The paths connecting distant salient points (such as the path between the tip of a finger and the tip of a toe of a human model) can then be used to identify (implicitly) necks that should contain good cuts.

\mypar{Approximation algorithms.}

In \secref{opt_approx}, we present an efficient approximation algorithm to the optimal \collar (i.e., best neck-cut) under certain (strong) conditions. This approximation algorithm gives us reasonable bounds for the total time complexity required, while also motivating the core heuristics used in the practical algorithm.

In \secref{algorithm}, we present our practical algorithm. Our algorithm uses the salient points to define paths that the neck-like cycle must cross. Elegantly, we can now mainly rely on shortest path algorithms with little preprocessing and focus on filtering cycles of interest.

Our practical algorithm is based on a few heuristics derived from the properties of bottleneck curves, and thus can produce good bottleneck curves in surface meshes. We discuss the performance of our algorithm and the viability of generating these curves in a real-time setting.  Unlike previous work, our algorithm avoids complex global computation or iterative smoothing.

Numerous examples of the output of our algorithm are provided at \url{https://neckcut.space} and discussed in \secref{evaluate}.

\section{Isoperimetry, bottleneck cuts, and salient points}
\seclab{background}

\subsection{Isoperimetry and bottlenecks}
\seclab{iso_necks}

\mypar{Isoperimetric problem on surfaces.}

The isoperimetric problem asks to determine a plane figure of maximum area, with a specified boundary length. This problem dates back to antiquity, but a formal solution was not provided until the 19th century.  It is known \cite{c-rcwi-58} that circles, and in higher dimensions balls, are the optimal shapes. Even in the plane, proving it was quite a challenge. For a planar closed curve $\sigma$, consider its \emph{isoperimetric ratio}:
\begin{equation*}
    \rho(\sigma) = \frac{ \areaX{\interX{\sigma}\bigr.} }{ \lenX{\sigma}^2 },
\end{equation*}
where $\interX{\sigma}$ is in the interior region bounded by $\sigma$, $\areaX{\interX{\sigma}\bigr.}$ is its area, and $\lenX{\sigma}$ is the length of $\sigma$. This ratio can be arbitrarily small (i.e., consider a wiggly shape that has a small area but a long boundary). The isoperimetric inequality states that this ratio is maximized for the disk, where it holds with equality. Namely, the isoperimetric inequality states that, for any closed planar curve $\sigma$, we have
\begin{math}
    \rho(\sigma) \leq  \tfrac{1}{4\pi}.
\end{math}

\smallskip%

On a finite surface (say of genus zero) in 3D, it is natural to try to compute a closed curve on the surface as short as possible that splits the surface area into two ``large'' parts. As a concrete example, consider the natural cycle in the base of a human finger -- it does not partition the surface (i.e., a human model) even remotely equally. And yet, it is intuitively a good neck-cut.

\mypar{Tightness.}

To overcome this for a surface $\model$ (say of genus $0$), we define a variant of the isoperimetric ratio.

\begin{definition}
    \deflab{tightness}%
    For a surface $\model$ in $\Re^3$ of genus zero, and a region $\patch \subseteq \model$, let the \emphi{tightness} of $\patch$ is the ratio
    \begin{equation*}
        \tightnessX{\patch}
        =%
        \frac{\min\bigl( \areaX{\patch}, \areaX{\smash{\opatch}} \bigr) }
        {\lenX{\partial b}^2},
    \end{equation*}
    where $\partial{\patch}$ is the \emphw{boundary} of $\patch$, $\areaX{\patch}$ is the \emphw{area} of $\patch$, the complement of $\patch$ is $\opatch = \model \setminus \patch$, and $\lenX{\partial{\patch}}$ denotes the length of $\partial \patch$.  In particular, for a close curve $\crD$, that splits the surface into two parts $\patch$ and $\opatch$, its \emphi{tightness} is $\tightnessX{ \crD } = \tightnessX{\patch}$.
\end{definition}

Here, the closed curve is $\partial \patch$, the patches generated by $\partial \patch$ are $\patch$ and $\opatch$, and its tightness is the isoperimetric ratio.  It is thus natural to ask for the patch $b \subseteq \model$ with maximum tightness.

\bigskip
\noindent
\begin{minipage}{0.85\linewidth}
    \begin{example}
        \exalab{sphere}%
        Consider the $\model$ to be a unit radius sphere in 3D. It is not hard to see that the maximum tightness is realizable by $\patch$ being the (say, top) hemisphere of $\model$, and $\partial \patch$ being the equator. In that case, $\areaX{\patch} = 2\pi$, and $\tightnessX{\patch} = 2\pi / (2\pi)^2 = \tfrac{1}{2\pi}$.  Intuitively, closed curves are ``interesting'' as far as being a neck-cut if their tightness is at least $\tfrac{1}{2\pi}$.  (Compare this to the disk in the plane, that has tightness $\tfrac{1}{4\pi}$.)
    \end{example}
\end{minipage}
\hfill
\begin{minipage}{0.13\linewidth}
    \includegraphics[width=0.99\linewidth]{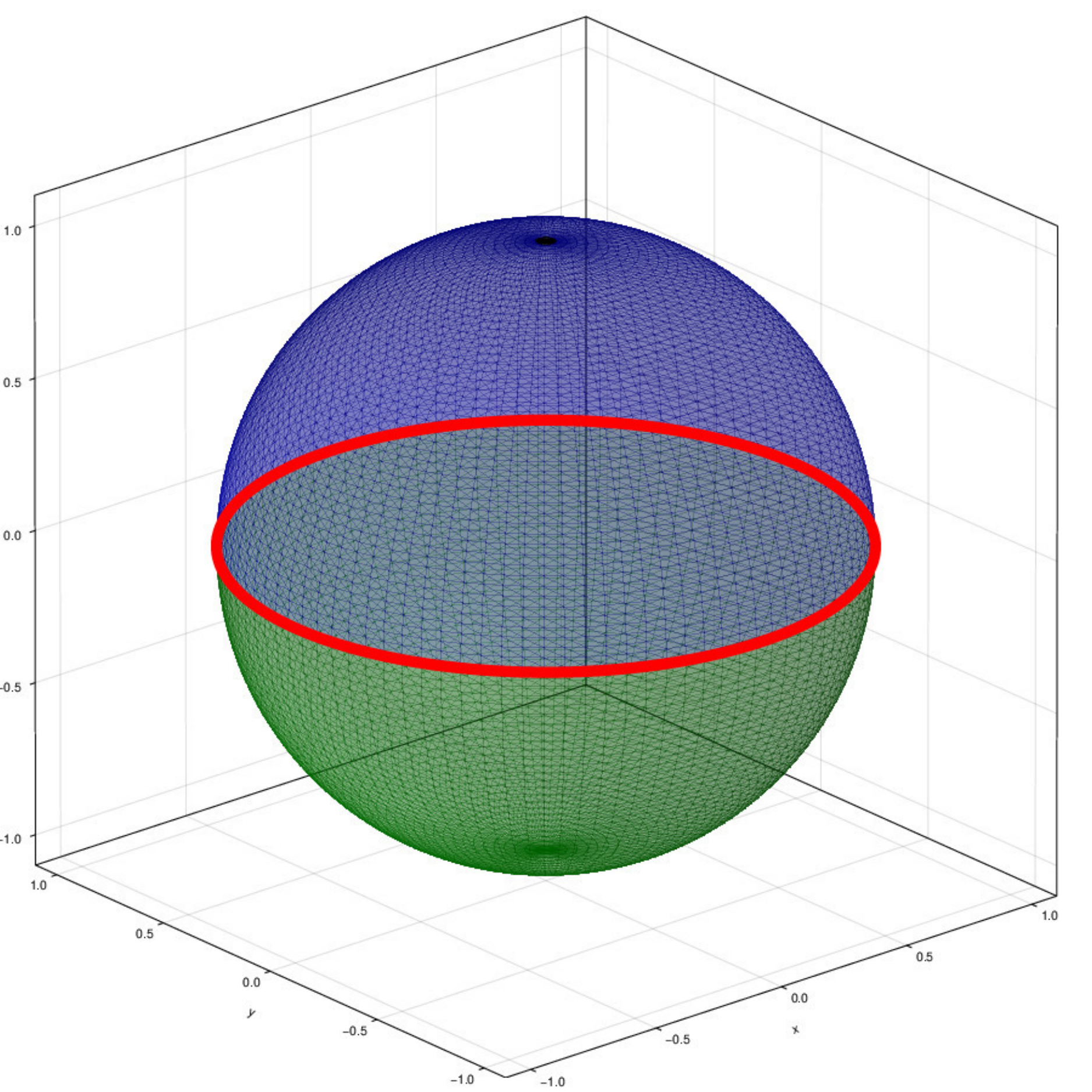}
\end{minipage}

\bigskip%

In general, proving that a specific patch $b \subseteq \model$ is the one realizing maximum tightness is a challenging problem, as the long history of the isoperimetric inequality testifies \cite{o-ii-78, cfgns.ea-ipsss-05, hhm-ips-99}. A bottleneck curve on a surface should have the property that it is short, while enclosing a large area on both sides (e.g., a neck of an hourglass). Thus, our proxy for finding a good bottleneck curve is going to be computing curves on a given surface that have high tightness.

\begin{problem}
    Given a surface $\model$, compute a region $\patch$ with tightness $\tightnessX{\patch}$ as large as possible.
\end{problem}

\begin{figure}[ht]
    \phantom{}\hfill%
    {\includegraphics[page=3]{figs/almost_geo}}%
    \hfill%
    \includegraphics[page=1]{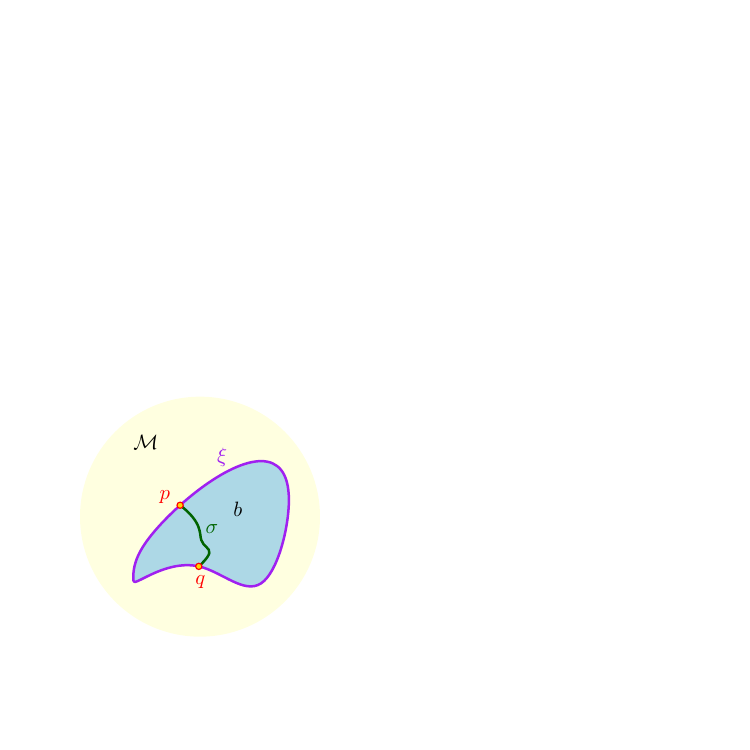} \hfill%
    \includegraphics[page=2]{figs/almost_geo} \hfill\phantom{}%

    \caption{Middle, Right: For some optimal bottleneck $\crA$, we consider the geodesic $\sigma$. Left: If the geodesic were to cross $\crA$ at $x$, it would be a contradiction, as shortcutting along $\crA$ would be shorter.}
    \figlab{opt_neck}
\end{figure}

\begin{definition}
    Let $\model$ be a surface of genus zero, and let $\crA$ be a cycle on $\model$.  Let $\patch$ be the region bounded by $\crA$ on $\model$.  The cycle $\crA$ is a \emphi{$\alpha$-bottleneck} of $\model$, if the tightness of $\patch$ is at least $\alpha$. The $\alpha$-bottleneck with maximum $\alpha$ on $\model$, is the \emphw{optimal bottleneck cut}, or simply a \emphi{\collar}.
\end{definition}

\subsection{Well behaved surfaces}
\seclab{well_behaved}

To provide some intuition for the approach used in \secref{algorithm}, we discuss a few properties of a \emph{well-behaved} surface.

\mypar{Slow-expansion on the surface.}

We are interested in surfaces such that their measure (i.e., area/volume) does not expand too quickly.  To this end, given a set $\sigma$ on a $\model$, its \emphw{$r$-expansion} is the region
\begin{equation}
    \sigma \oplus r
    =%
    \Set{ p \in \model}{\dMY{p}{\sigma} \leq r},
    \eqlab{expansion}
\end{equation}
where $\dMC$ is the geodesic distance along $\model$.

\begin{definition}
    A model $\model$ is \emphi{$\tau$-expanding}, if for any curve
    $\sigma \subseteq \model$, and any $r \geq 0$, we have that
    \begin{math}
        \areaX{\sigma \oplus r} \leq \tau( \lenX{\sigma} r + r^2)
    \end{math}
    and
    \begin{math}
        \lenX{\bd (\sigma \oplus r) } \leq \tau (\lenX{\sigma} + r).
    \end{math}
    The minimum such $\tau$ is the \emph{expansion} of $\model$, denoted by $\Expansion$.
\end{definition}

\begin{example}
    In the plane, Steiner inequality \cite{g-t-04} implies that for any curve $\sigma$ we have
    \begin{math}
        \areaX{\sigma \oplus r} \leq \pi r^2 + 2r \lenX{\sigma}
    \end{math}
    and
    \begin{math}
        \lenX{\partial(\sigma \oplus r)}
        \leq
        2 \lenX{\sigma} + 2\pi r.%
    \end{math}
    Thus, the plane is $2\pi$-expanding.
\end{example}

\begin{definition}
    A simple cycle $\sigma$ is \emphi{contractible} if one can continuously morph $\sigma$ to a single point.  A portion $\RP \subseteq \model$, and a cycle $\sigma \subset \RP$, the cycle is \emphi{$\RP$-contractible}, if it can be contracted to a point, while staying inside $\RP$.
\end{definition}

Note that while any cycle $\sigma$ on the original surface $\model$ is contractible, if $\model$ has genus zero, it might not be $\RP$-contractible -- for example, if $\RP$ is the result of taking $\model$ and creating two punctures on both sides of $\sigma$.  Intuitively, tight cycles are not locally contractible -- one has to go far to be able to collapse them to a point.

\begin{definition}
    For $\RP \subseteq \model$, and an $\RP$-contractible cycle $\sigma \subseteq \RP$, let $\patch_\sigma \subseteq \RP$ be the portion of $\model$ bounded by a closed curve $\sigma$ (the other portion of $\model$ bounded by $\sigma$ might contain portions outside $\RP$).  If $\RP = \model$, let $\patch_\sigma$ denote the smaller area patch (out of the two patches) induced by $\sigma$ on $\model$.  The region $\RP$ is \emphi{$\alpha$-well-behaved} if for all $\RP$-contractible cycles $\sigma \subset \RP$, we have that $\areaX{ \patch_\sigma } \leq \alpha\lenX{\sigma}^2$.
\end{definition}

\begin{remark}
    Consider a region $\RP \subseteq \model$, where $\model$ is $\tau$-expanding, such that any $\RP$-contractible cycle $\sigma$ in it, is $\bigl((\sigma \oplus r) \cap \RP\bigr)$-contractible, where $r = \tau \lenX{\sigma}$. Then, the $\tau$-expansion implies that $\areaX{ \patch_\sigma } \leq \areaX{\sigma \oplus r} \leq \tau( \lenX{\sigma} r + r^2) \leq 2 \tau^3 \lenX{\sigma}^2)$. Namely, $\RP$ is $2\tau^3$-well behaved.
\end{remark}

\subsection{Salient points to a bottleneck}
\seclab{salient}

In the following, we assume the given surface \model is triangulated, has genus $0$, and it has a useful \collar (i.e., $\alpha$-tight for a ``large'' $\alpha$). In addition, we assume $\model$ is $\tau$-expanding, where $\tau$ is some small constant.  Let $\neck$ denote this optimal $\alpha$-bottleneck of $\model$. The cycle $\neck$ breaks $\model$ into two regions $\patch$ and $\opatch$. Let $\salientY{\neck}{\patch}$ be the point furthest away from $\neck$ on $b$. Formally, we define
\begin{equation*}
    \salientY{\neck}{\patch}%
    =%
    \arg \max_{ p \in b} \dMY{p}{\neck}
    \qquad\text{where}\qquad
    \dMY{p}{\neck}  = \min_{q \in \neck} \dMY{p}{q}.
\end{equation*}
Such points are \emphi{salient}, and they are far from the bottleneck if the surface is well behaved.

\begin{lemma}[salient points are far]
    For $s = \salientY{\neck}{\patch}$, we have that $\dMY{s}{\neck} \geq \lenX{\neck}$, if $\alpha \geq 4 \tau$.
\end{lemma}
\begin{proof}
    Let $\ell = \lenX{\neck}$. By $\neck$ being an $\alpha$-bottleneck, we have that
    \begin{equation*}
        \areaX{\patch}
        \geq
        \alpha \lenX{\neck}^2.
    \end{equation*}
    On the other hand, for $r = \dMY{s}{\neck}$, we have
    \begin{math}
        b%
        \subseteq%
        \neck \oplus r.
    \end{math}
    By the $\tau$-expansion of $\model$, we have
    \begin{equation*}
        \areaX{\patch}
        \leq%
        \areaX{\neck \oplus r}
        \leq%
        \tau(  \lenX{\neck} r + r^2).
    \end{equation*}
    Thus, we have
    \begin{math}
        \alpha \lenX{\neck}^2 \leq \tau( \lenX{\neck} r + r^2) \leq%
        \tau( \lenX{\neck}/2 + r)^2.
    \end{math}
    This implies that
    \begin{equation*}
        \frac{\alpha}{\tau} \lenX{\neck}^2
        \leq
        (  \lenX{\neck}/2 + r)^2
        \quad\implies\quad%
        r
        \geq
        \pth{\sqrt{\frac{\alpha}{\tau}} - \frac{1}{2}} \lenX{\neck}
    \end{equation*}
    as $\alpha / \tau \geq 4$.
\end{proof}

\section{Approximating the optimal \collar}
\seclab{opt_approx}

\subsection{Identifying the neck where the \collar lies}

Consider the easy case, that not only is there a good \collar, but this collar is stable, in the sense that one can slide it up and down the ``neck'' and the quality of the \collar remains relatively the same, see \figref{ostrich}.

\begin{figure}[ht]
    \centering
    \includegraphics[width=0.9\linewidth]{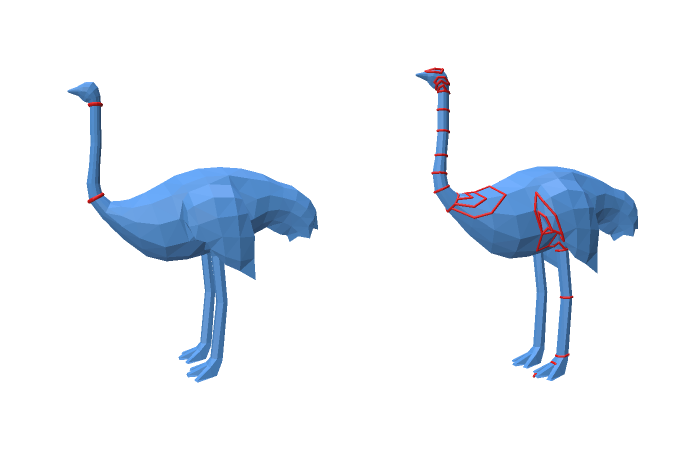}
    \caption{Left: A long neck with a stable collar. Right: Every geodesic cycle from the beak to the left foot.}
    \figlab{ostrich}
\end{figure}

\begin{definition}
    Let $\Path$ be a ``long'' shortest path on $\model$ with endpoints $s$ and $t$.  For every point $p \in \Path$, consider the shortest path $\spathX{p}$ from $p$ to itself, if we were to cut the surface $\model$ along $\Path$, and the path $\sigma$ has to connect $p$ to its copy. The closed curve $\spathX{p}$ is a \emphi{lasso} if
    \begin{equation*}
        \lenX{\spathX{p}}
        <
        \max\bigl(\dMY{p}{s}, \dMY{p}{t}\bigr),
    \end{equation*}
    and is denoted by $\spathX{p}$.
\end{definition}

Intuitively, a lasso is a ``short'' closed curve connecting $p$ to itself, which is shorter than going from $p$ to itself by going along the cut formed by $\Path$.

\begin{example}
    Let $\model$ be the surface of the following solid -- connect two large disjoint balls by a thin and long cylinder (i.e., a dumbbell) -- see \figref{dumbbell}. Consider the cylinder portion of the surface -- it forms a natural neck, and let $\RP$ denote it. Any curve going around the neck is not contractible on the neck, while any closed curve $\sigma$ that is contractible on the neck, is going to have area $O( \lenX{\sigma}^2)$. That is, the neck is $O(1)$-well behaved.
\end{example}
\begin{figure}
    \centering%
    \includegraphics[width=0.4\linewidth]{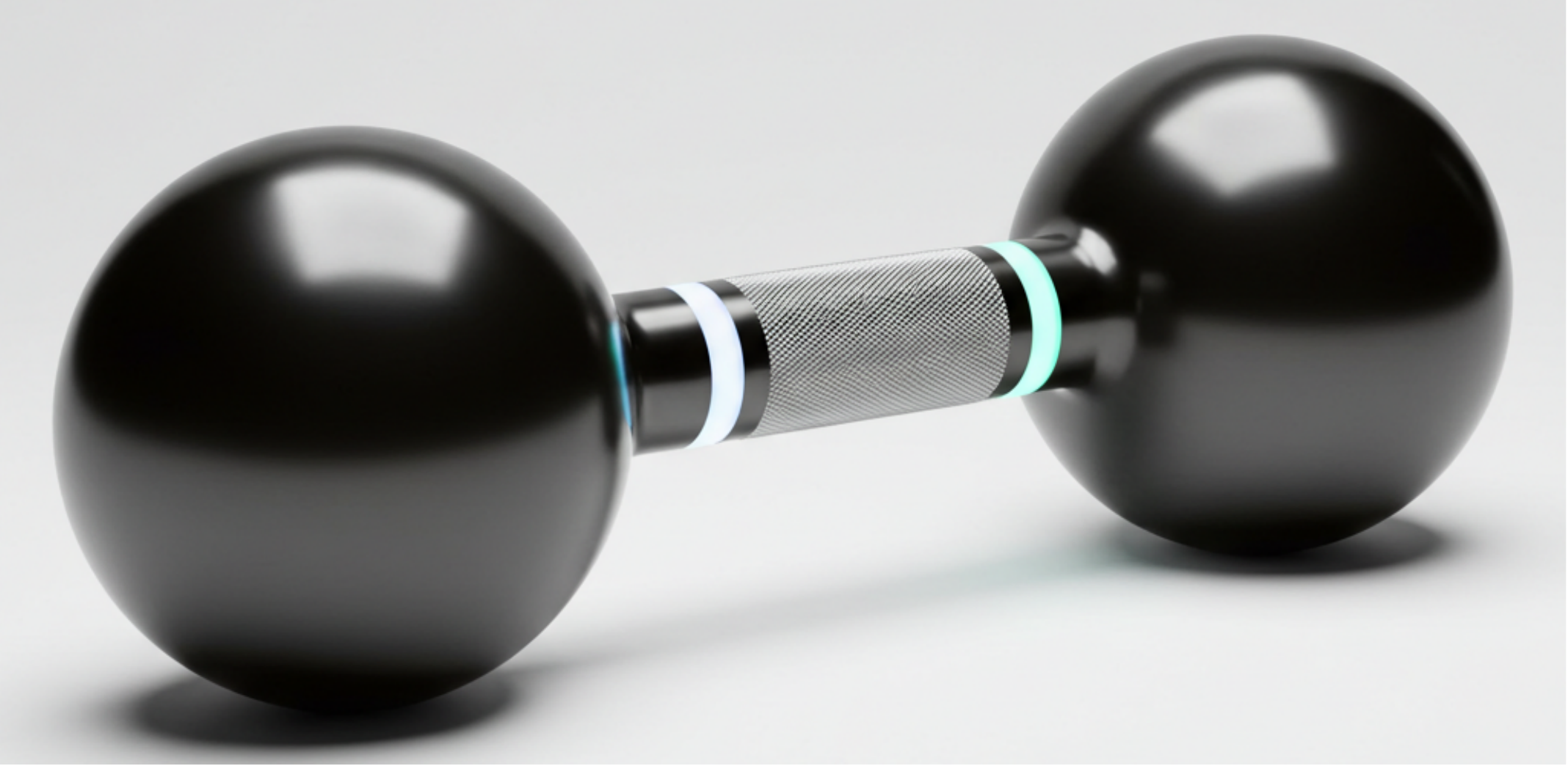}
    \caption{}
    \figlab{dumbbell}
\end{figure}

\begin{observation}
    Two lassos defined using the same base path $\pi$ can not cross each other.
\end{observation}

\begin{definition}
    Consider two lassos $\tau_1, \tau_2$ defined using a base path (which is a shortest path) $\pi$.  The \emphi{neck} $\neckC = \neckY{\tau_1}{\tau_2}$ is the area on the surface $\model$ lying between $\tau_1$ and $\tau_2$.  Such a region is a \emphi{$\beta$-neck} if it is $\beta$-well-behaved, for some $\beta > 0$.
\end{definition}

\begin{figure}[ht]
    \centerline{%
       \includegraphics{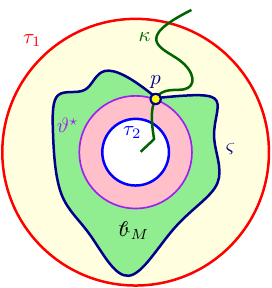}%
    }%
    \caption{}%
    \figlab{cycle}
\end{figure}

\begin{lemma}
    \lemlab{good_lasso}%
    Let $s,t$ be two points on $\model$, and let $\Path$ be the shortest path connecting them. Let $\tau_1,\tau_2$ be two lassos defined by points on $\Path$, and let $\neckC = \neckY{\tau_1}{\tau_2}$ be the induced $\alpha$-neck. Furthermore, assume that the optimal collar $\optCollar$ is contained in $\neckC$, and its tightness $\beta \geq 8\alpha$. Then, there exists a point $p \in \Path \cap \neckC$, such that its lasso $\crB = \spathX{p}$ is $\beta(1- \tfrac{4\alpha}{\beta})$-tight.
\end{lemma}

\begin{proof}
    Let $p$ be any point in $\Path \cap \optCollar$ as the base point for the construction, and let $\GLasso =\spathX{p}$ be the associated lasso.

    Assume, for now, that the lasso $\GLasso$ only intersects $\optCollar$ at $p$ (as in \figref{cycle}).  In that case, $\GLasso$ and $\optCollar$ are homotopic, and let $\patchM$ be the area in between them.  By the $\alpha$-behaveness of $\neckC$, and since $\partial \patchM$ is $\optCollar \cup \GLasso$ (and is $\neckC$-contractible), we have that
    \begin{equation*}
        \areaX{\patchM}
        \leq%
        \alpha (\lenX{\optCollar} + \lenX{\GLasso})^2%
        \leq%
        4 \alpha \lenX{\optCollar}^2,
    \end{equation*}
    as $\lenX{\GLasso} \leq \lenX{\optCollar}$ -- indeed, $\optCollar$ is a candidate for the shortest path connecting $p$ to itself ``around'' $\Path$, but $\GLasso$ is the shortest one.

    For a closed curve $\crD$, let $\regX{\crD}$ and $\oregX{\crD}$ be the two parts of $\model$ bounded by $\crD$.  Observe that
    \begin{equation*}
        B
        =%
        \min\bigl( \areaX{\regX{\GLasso}},\areaX{\smash{\oregX{\GLasso}}}\bigr)
        \geq%
        \min\bigl( \areaX{\regX{\optCollar}},
        \areaX{\smash{\oregX{\optCollar}}} \bigr)
        - \areaX{\patchM}.
    \end{equation*}
    The tightness of $\GLasso$ is thus
    \begin{math}
        \displaystyle%
        \tightnessX{\GLasso}
        =%
        \frac{B}
        {\lenX{\GLasso}^2}
        \geq
        \frac{\min\bigl( \areaX{\patch}, \areaX{\smash{\opatch}} \bigr)
           - \areaX{\patchM}}{\lenX{\optCollar}^2}
        \geq
        \tightnessX{\optCollar} - 4 \alpha.
    \end{math}

    \begin{figure}[ht]
        \centerline{%
           \includegraphics[page=2]{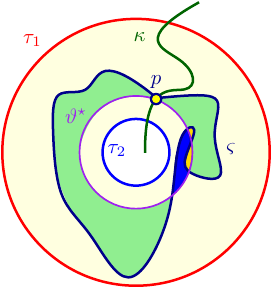}%
        }%
        \caption{}%
        \figlab{cycle_in_out}
    \end{figure}

    The slightly harder case is when $\GLasso$ and $\optCollar$ have several intersections. In that case, the set $\GLasso \cup \optCollar$ forms an arrangement -- the ``inner'' region/face denoted by $\patch_s$ and the outer face denoted by $\patch_t$. So consider all the other faces $f_1, \ldots, f_k$ in this arrangement. These faces are all contractible disks. Let $\ell_i$ be the boundary of the $i$\th face, for $i=1,\ldots, k$. Observe that every edge $e$ of $\GLasso$ or $\optCollar$ contributes at most $2\lenX{e}$ to the total lengths of these boundary faces. Thus, we have
    \begin{math}
        \sum_{i=1}^k \ell_i
        \leq%
        2(\lenX{\GLasso} + \lenX{\optCollar})
        \leq
        2\lenX{\optCollar}.
    \end{math}
    Setting $\patchM = (\regX{\optCollar} \setminus \regX{\GLasso} )\cup (\regX{\GLasso} \setminus \regX{\optCollar} )$ to be region that is the symmetric difference between $\regX{\optCollar}$ and $\regX{\GLasso}$, we have
    \begin{equation*}
        \areaX{\patchM}
        \leq%
        \sum_{t=1}^k \areaX{f_t}
        \leq%
        \sum_{t=1}^k \alpha \ell_t^2
        =%
        \alpha \sum_{t=1}^k \ell_t^2
        \leq%
        \alpha \Bigl( \sum_{t=1}^k \ell_t \Bigr)^2
        \leq%
        4 \alpha  \lenX{\optCollar}^2.
    \end{equation*}
    The claim now follows from the argument above, and observing that
    \begin{math}
        \tightnessX{\GLasso}%
        \geq%
        \tightnessX{\optCollar} - 4 \alpha%
        =%
        \beta \bigl( 1 - \frac{4\alpha}{\beta} \bigr).
    \end{math}
\end{proof}

\begin{corollary}
    In the settings of \lemref{good_lasso}, if the tightness $\beta$ of the optimal collar on an $\alpha$-neck $\neckC$ is $\geq 4\alpha / \eps$, for $\eps \in (0,1)$, then there is lasso on $\neckC$ of tightness $\geq (1-\eps) \beta$. Namely, the lasso has tightness $\eps$-close to optimal.
\end{corollary}

\subsection{The algorithm for computing the collar}
\lemref{good_lasso} implies that if the optimal tightness is much bigger than the $\alpha$ (the well-behaveness of the neck), then one can efficiently compute a collar with tightness close to optimal.  Significantly, such a lasso is efficiently computable.

We outline here the basic idea -- our purpose is to present a polynomial-time approximation algorithm. To this end, we guess the points $s$ and $t$ of \lemref{good_lasso}, and compute the shortest path $\Path$ between them. Next, we guess the two points, $x,x' \in \Path$, defining the lassos bounding the optimal collar $\optCollar$, and we compute these two lassos. We compute the region $\neckC$ bounded in between $\tau_1$ and $\tau_2$, and verify that it indeed has the topology of a sleeve (this can be done by example by computing its Euler characteristic, and verifying that $\tau_1$ and $\tau_2$ cover all the boundary edges of this patch. Now, one can try to compute the shortest path around the neck for each vertex $v \in \Path \cap \neckC$, and explicitly determine its tightness. The maximum one found is the desired approximation. If the model has size $n$, the running time of this algorithm is $O(n^5)$. We thus get the following.

\begin{theorem}
    \thmlab{o_collar}%
    Let $\model$ be a triangulated surface in $3D$ with genus $0$ and $n$ vertices.  Assume the optimal collar $\optCollar$ on $\model$ lies on an $\alpha$-neck $\neckC$ that is induced by a shortest path $\Path$, and two lassos $\tau_1, \tau_2$ (see \lemref{good_lasso}). Furthermore, the tightness $\tightnessX{\optCollar} \geq 8\alpha$. Then, one can compute, in $O(n^5)$ time, a closed curve $\crB$ such that $\tightnessX{\crB} \geq \tightnessX{\optCollar} - 4\alpha \geq \tightnessX{\optCollar}\!/2$.
\end{theorem}

\begin{remark}
    The above algorithm inspires our practical algorithm, which achieves sub-quadratic running time by avoiding the need to guess all pertinent information. Thus, we had not spent energy on improving the running time of \thmref{o_collar}. In particular, the stated running time should be taken as evidence that the optimal collar can be approximated efficiently under certain conditions.
\end{remark}

\section{A practical algorithm}
\seclab{algorithm}

We present here a practical algorithm that utilizes a few heuristics to locate optimal bottleneck curves while maintaining fast performance.  This algorithm is an approximation and may not find all optimal curves. We assume the input is a surface model of genus zero.

In the first stage, the algorithm locates salient points that lie on the tips of features. In the second stage, the algorithm takes paths between salient points and searches for bottleneck curves. We next describe these in detail. In \secref{evaluate}, we discuss how this algorithm performs on various real-world inputs.

\subsection{Computing salient points}

We aim to identify distant pairs of \emph{salient points} to search for bottleneck curves. Usually, these points represent the tips of various features on a mesh (such as the tip of a spike, finger, or head) but do not necessarily lie on the convex hull of the mesh.  Finding exact salient points can be difficult and time-consuming. We use a standard method to locate points that are near ``true'' salient points, but are sensitive to the starting location of our search.  Previous work has focused on identifying salient points and utilizing them in mesh decomposition \cite{zh-dpmbm-04}. Other methods have also utilized feature points and surface methods to perform mesh decomposition \cite{ddu-smsbg-23, gf-cs3m-09, klt-msufp-05, dgg-ssmfd-03}. We use shortest path algorithms to locate these salient points, which act as reasonable estimates for the exact salient points that would be used to find bottleneck cuts.

We emphasize that in the following, we treat the lengths of the edges of the mesh as forming a graph, and our algorithm would work on this triangulated graph. The length of the edges is the Euclidean distance between their endpoints.

\paragraph{Step I: Salient point via approximate diameter.}
The algorithm first computes a point on the mesh that is part of a $2$-approximation of the diameter of the mesh: \medskip%
\begin{compactenumI}
    \item The algorithm picks an arbitrary point on the graph $s$.

    \smallskip%
    \item The algorithm computes the shortest-path tree $T_s$ from $s$ using Dijkstra. Let $u$ be the leaf with the maximum distance from $s$.

    \smallskip%
    \item The algorithm computes the shortest-path tree $T_u$ from $u$, again using Dijkstra. Let $v$ be the leaf with maximum distance from $u$. The pair $u$ and $v$ is a $2$-approximation to the diameter of the mesh (when restricted to paths along the edges)
\end{compactenumI}
\medskip%
A vertex $x$ is a \emphi{salient point} if it is a local maximum of its neighbors with respect to distance to $u$. That is,
\begin{equation}
    \forall y \in \Gamma(x) \qquad \dMY{u}{y} < \dMY{u}{x},
    \eqlab{local_condition}
\end{equation}
where $\Gamma(x)$ is the \emphi{neighborhood} of $x$ (i.e., set of all vertices adjacent to $x$). Because $v$ is the furthest point from $u$, all its neighbors must have a shorter distance, and thus it is the first salient point computed.

\paragraph{Stage II: Collecting candidates for salient points.}
We now extend this idea further.  For every leaf $x$ on $T_u$, we check all its neighbors on the mesh and mark it as a salient point if it is a local maximum of the distance function from $u$. That is, the algorithm checks if $x$ complies with \Eqref{local_condition}.  Let $C$ denote the set of candidate points picked.

\paragraph{Stage III: Filtering the candidates.}
Depending on the quality of the mesh, some further filtration of salient points may be necessary. Extremely noisy surfaces can create several local maximum points close together. In these cases, we want to eliminate any salient point within a user-selected (hop) distance $r$ from any salient point. To this end, the algorithm performs an $r$-depth breadth-first search from each salient point of $C$, removing from this set any point that fails to be a local maximum within its $r$-hop neighborhood. In practice, $r$ is a small constant, and it is sufficient to handle small perturbations in a mesh.  See \figref{salient} for an example.

\begin{figure}
    \centering
    \includegraphics[width=0.8\linewidth]{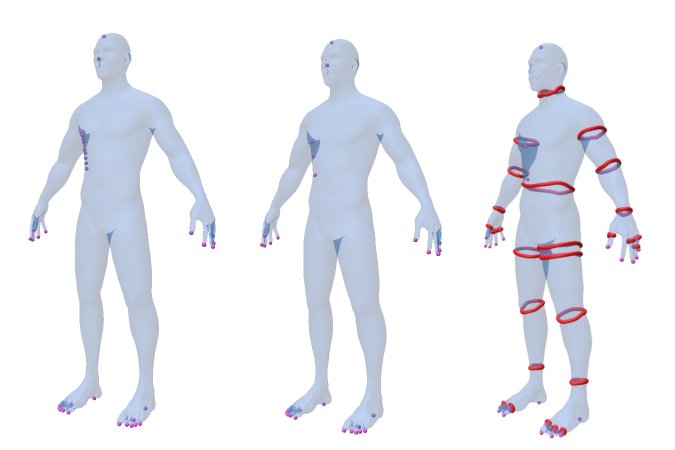}
    \caption{The salient points of a human mesh, with no filtering (left), $r=10$ (middle), and $r=20$ (right). The bottleneck curves are shown in the right mesh.}
    \figlab{salient}
\end{figure}

\begin{observation}
    For $r$ a constant, the above algorithm runs in $O(n \log n)$ time, assuming the vertices' degrees in the mesh are bounded. Indeed, it involves running the Dijkstra algorithm several times, and doing local \BFS of depth at most $r$ for each candidate salient point.
\end{observation}

\subsection{Finding bottleneck curves}

\subsubsection{Generating candidate neck-cuts from a shortest path}

Consider a shortest path $\pi$. The algorithm runs a shortest-path algorithm from some vertex $v \in \pi$ to itself, where no vertices on $\pi$ can be used, thus preventing the shortest-path algorithm from crossing $\pi$. That is, the algorithm ``cuts'' the surface along $\pi$, computing shortest paths ``around'' $\pi$ from $p$ to its copy on the other side of $\pi$. This path is a geodesic (with respect to $\pi$) cycle which includes $v$. This process generates an ordered sequence of cycles $\Pi$, along $\pi$, base on their base vertex $p$.  The task is to select cycles with a local maximum tightness. To compute the tightness requires computing the area bounded by each cycle $\crA \in \Pi$, and selecting local maxima within $\Pi$. On genus zero objects, we can compute the area between every adjacent pair of cycles (e.g. $\crA_i, \crA_{i+1} \in \Pi$ for each $i$), by doing \BFS on the area in between the two curves, and then use prefix sums to compute the area bounded by one cycle efficiently. For a cycle $\crA$, if its tightness is the local maximum among its neighbors, then it is likely a good bottleneck.

Cycles near boundaries and salient points might be of low quality.  Likewise, cycles that bound a small amount of area, whilst also being short, can lead to other cycles which appear to be local maxima among their neighbors. To that end, we refer to the calculation in \exaref{sphere}, and filter out any cycle with a computed tightness less than $\frac{1}{2\pi}$ (with some tolerance due to inaccuracies with triangulated meshes). These cycles are not representative of a large enough bounded area, and thus do not make good bottleneck curves.

\subsubsection{Computing a skeleton:  A small set of shortest paths}

To use the above algorithm for computing neck-cuts, we need to feed it the shortest paths to run on. A natural but naive approach would be to try all pairs of shortest paths between pairs of salient points in the set $C$. To reduce the number of paths, we compute a ``cheap'' skeleton of the salient points. A natural candidate for such a skeleton is the Steiner tree of the points of $C$ on $\model$, but unfortunately, computing the exact Steiner tree is \NPComplete. To avoid this, we use a Heuristic to construct such a tree. The algorithm starts with the approximate diametrical pair $u$ and $v$ and connects them by the shortest path. The algorithm then adds the remaining points to the constructed graph iteratively, connecting the $i$\th salient point to its nearest-neighbor on the tree constructed so far, adding this connecting path to the collection of shortest-paths computed; their union is the computed skeleton. We then use these paths with the above procedure to generate the good neck-cuts.

\begin{remark}
    The algorithm described here does not have any guarantees as far as the quality of approximation it provides. One can be a bit more careful and provide $2$-approximation by being slightly more careful. Indeed, assume the set $C = \{s_1,\ldots, s_k \}$. In the $i$\th iteration, for $i=2, \ldots k-1$, the algorithm computes the closest point in $C \setminus C_i$ to $C_i = \{ c_1, \ldots, c_i\}$. This step can be done by a single run of Dijkstra, and let $c_{i+1}$ denote this point. The algorithm then computes the shortest path from $c_i$ to $T_{i-1}$ (i.e., the tree constructed in the first $i-1$ iterations), and adds this path to this tree, to form the new tree $T_{i}$. This algorithm can be viewed as running Prim's algorithm on the induced complete graph over $C$ (where the weights are the shortest path distances between the corresponding points on the model). The constructed tree $T_k$ is clearly no heavier than this \MST, which in turn is a $2$-approximation to the optimal Steiner tree. The running time of this algorithm is $O( k n \log n)$. We have not implemented this exact variant, since it does not really seem to matter in practice.
\end{remark}

\subsection{Running time analysis}
\seclab{runtime}

Computing the set of candidates takes $O(n \log n)$ time, for running shortest paths, and the $r$-depth filtering. Let $k=|C|$ and $|K|$ be the number of vertices in the skeleton. Computing the skeleton itself takes $O(k n\log n)$ time. The bulk of our time results from computing the cycles along the paths, which takes $O(|K|n \log n)$ time. In practice, this is much faster, since we are running an $st$-shortest path algorithm, which terminates early.

There are additional optimizations in practice that can be employed that would further speed up the cycle searching. Mainly, computing the cycles along each path is not reliant on data from the other paths in $K$. Namely, with parallelism, each path can be computed independently, thus bounding the runtime based on the number of available cores and the longest path within the skeleton.

\section{Evaluation}
\seclab{evaluate}

\newcommand{\CPP}{\texttt{C++}\xspace}

We implemented our algorithm in \CPP, using the Polyscope \& Geometry Central \cite{s.ea-p-19,geometrycentral} libraries. These libraries provide a halfedge data structure, with standard traversal algorithms. All timings were measured on a single thread of an Intel i7-14700K 3.4GHz CPU.  We tested our algorithm on \textit{Benchmark for 3D Mesh Segmentation Dataset} \cite{chen-meshseg}, along with additional meshes. All of our results from this dataset, along with the code, can be viewed on \href{https://neckcut.space/dist/}{meshcuts.space}. We have selected a few to feature in this paper, as shown in \tabref{inputs}. We measured the runtime of all genus zero models, as shown in \figref{timing_data}. We attempted to use the code from \cite{at-fcnf-23} for comparison. However, we were unable to reproduce a working program, despite repeated efforts.

\begin{table}[ht]

    \centering
    \begin{tabular}{|l|l|r|}
      \hline
      \textbf{Input}
      & \textbf{Description}
      & \textbf{Faces}
      \\ \hline
      \rowcolor[HTML]{D0D0D0}
      7
      & Public Domain: Human 1
      & $48,918$
      \\ \hline
      2
      & MSB/Stanford: Armadillo
      & 50,542
      \\ \hline
      \rowcolor[HTML]{D0D0D0}
      3
      & MSB: Octopus
      & 28,248
      \\ \hline
      4
      & MSB: Ant
      & 13,696
      \\ \hline
      \rowcolor[HTML]{D0D0D0}
      5
      & MSB: Horse
      & 11,072
      \\ \hline
      6
      & MSB: Hand
      & 3,026
      \\ \hline
      \rowcolor[HTML]{D0D0D0}
      1
      & MSB: Human 2
      & 11,258
      \\ \hline
      8
      & Stanford: Bunny
      & $69,630$
      \\ \hline
    \end{tabular}
    \caption{Selected inputs from our testing. MSB models are from \cite{chen-meshseg}, Stanford models are from \cite{s-s3sr-14}}
    \tablab{inputs}
\end{table}

Our implementation still has some optimizations to be yet implemented, as described in \secref{runtime}. The timings described in \tabref{timing} are single-threaded operations. We lazily compute the bounded area in this implementation, only running the prefix-sum method per path, rather than the entire skeleton at once. However, further optimizations would not have a significant impact on the runtime. The results from these models can be seen in \figref{human_arm}, \figref{octopus}, and \figref{hands_it}.

When deciding which cycles to display, we chose cycles whose tightness is a local maximum among a window of five cycles. From the discussion in \secref{opt_approx}, a neck is defined by two collars, with some optimal collar lying within the neck. In practice, we observe that several cycles all reach a maximum tightness as a group, since these neck-like surfaces are usually well-behaved. In dense meshes, such a small exclusion window would result in several cycles with similar tightness reported together. Alternate implementations can choose to report all or some of these cycles, but in either case, the same neck-like feature is identified.

\begin{figure}
    \centering
    \includegraphics[width=0.65\linewidth]{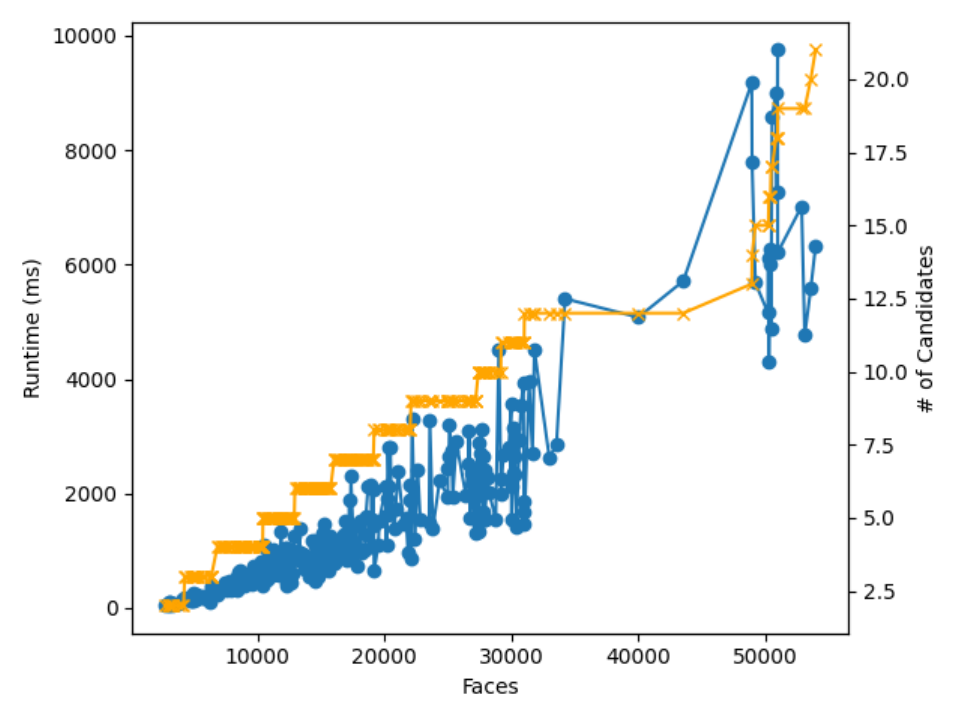}
    \caption{Runtime of models from \cite{chen-meshseg}. Runtime is plotted with \textit{blue} dots. The number of candidates is plotted with \textit{orange} crosses.}
    \figlab{timing_data}
\end{figure}

\begin{figure}[ht]
    \centering
    \includegraphics[width=0.9\linewidth]{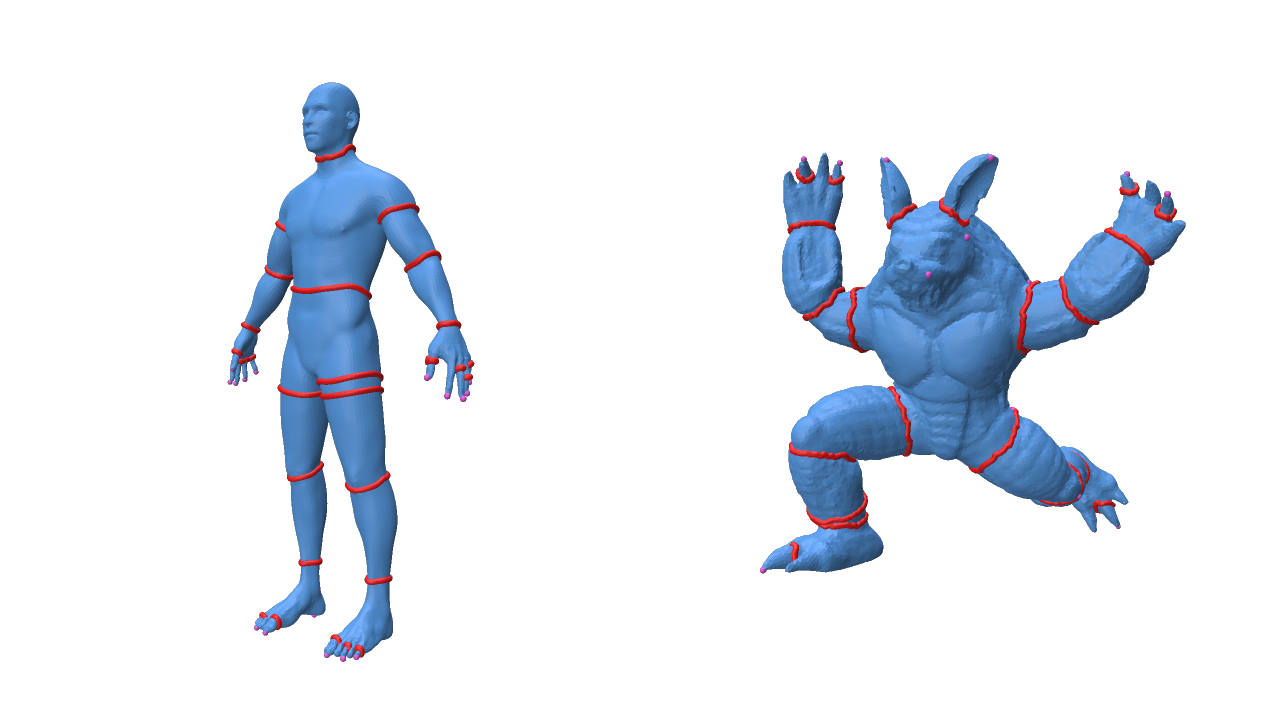}
    \caption{Inputs 1 \& 2}
    \figlab{human_arm}
\end{figure}

\begin{figure}[ht]
    \centering
    \includegraphics[width=0.9\linewidth]{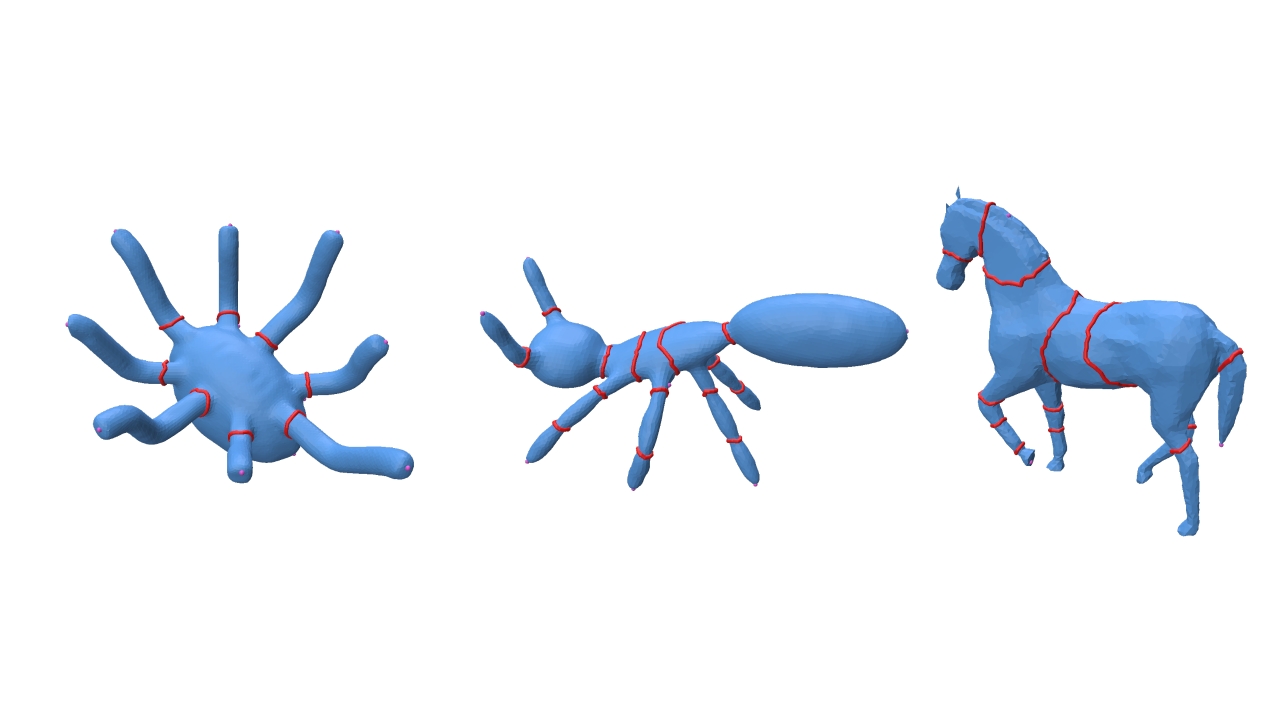}
    \caption{Inputs 3, 4, 5}
    \figlab{octopus}
\end{figure}

\begin{figure}[ht]
    \centering
    \includegraphics[width=0.9\linewidth]{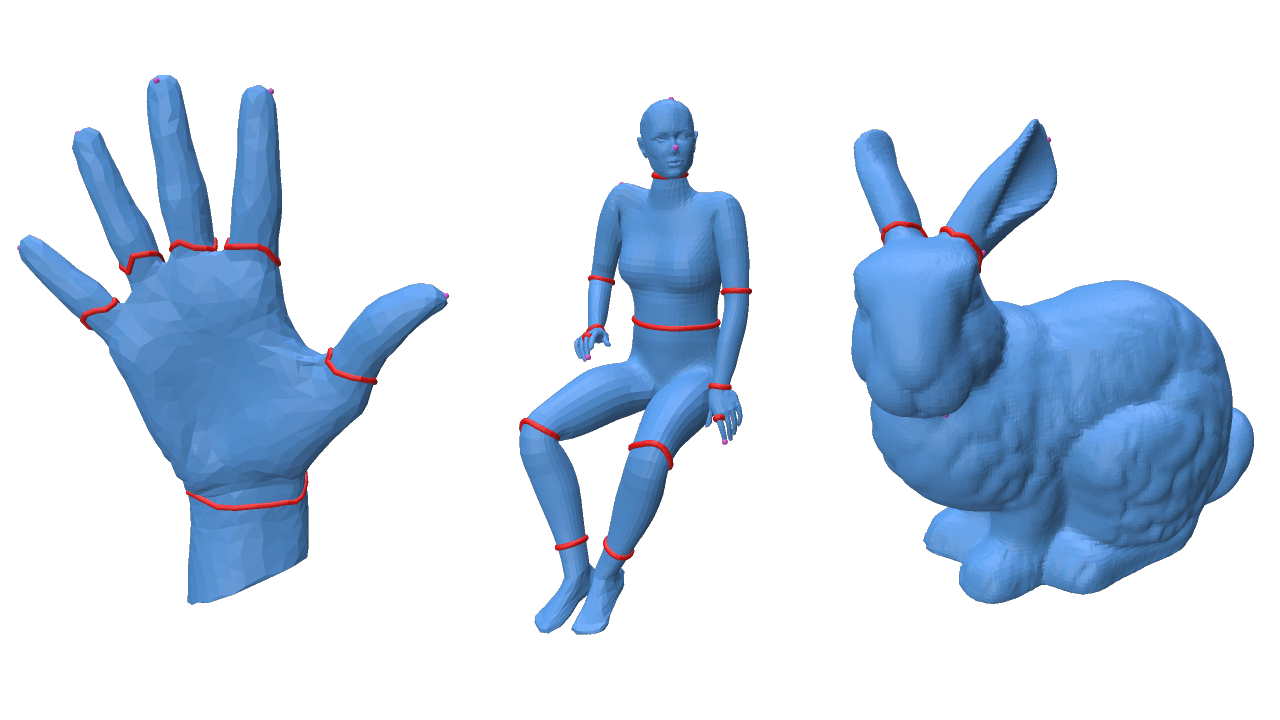}
    \caption{Inputs 6, 7, 8}
    \figlab{hands_it}
\end{figure}

\begin{table}[ht]
    \centering
    \begin{tabular}{|c|r||r|r|r|r|}
      \cline{3-6}
        \multicolumn{2}{c||}{}
      &
        \multicolumn{4}{c|}{\textbf{Runtime in ms}}
      \\ \hline
      {Input}
      & {\#}
      & Salient
      & Cycles
      & Tightness
      & {Total}
      \\
      \hline
      \hline
      1
      & 21
      & $182$
      & $4,252$
      & {398}
      & $4,833$
      \\
      \hline
      {2}
      & 22
      & 185
      & $5,428$
      & {413}
      & $6,026$
      \\
      \hline
      {3}
      & 12
      & {89}
      & $1,352$
      & $129$
      & $1,570$
      \\
      \hline
      {4}
      & 11
      & {43}
      & {578}
      & {57}
      & 678
      \\
      \hline
      {5}
      & 9
      & {35}
      & {516}
      & {37}
      & 588
      \\
      \hline
      {6}
      & 6
      & {11}
      & {98}
      & {6}
      & 115
      \\
      \hline
      {7}
      & 7
      & {42}
      & {496}
      & {28}
      & 567
      \\
      \hline
      {8}
      & 9
      & {226}
      & $8,151$
      & {230}
      & $8,607$
      \\\hline
    \end{tabular}
    \caption{Runtimes of the algorithm on selected inputs. The runtime is sensitive to the number of discovered points, which is also disclosed. Here, $r=20$ was used for salient point filtering. The column $\#$ is the number of salient points computed. the \emph{Salient} column is the time to discover the set of salient points, and connect them into a skeleton. The \emph{Cycles} column is the time for discovering every cycle along the skeleton. The \emph{Tightness} column is the time for area computation, tightness computation, and cycle filtering.}
    \tablab{timing}
\end{table}

\section{Conclusions}
\seclab{discuss}

We proposed a new definition for neck-like features and the curves that defined them. We also presented an approximation and practical algorithm to detect these bottleneck curves on real-world meshes. We believe our method has improvements over previous work \cite{tvd-tdmse-06, h-c-05, at-fcnf-23}, while also being extremely simple to implement. For future work, we plan to explore using this algorithm in other applications. One possible application is to use this algorithm as the seeds for \cite{xhf-ecegl-12}, rather than user-defined cutting planes, to discover exact geodesics from our output cycles.

\FloatBarrier

\BibTexMode{%
}%
\BibLatexMode{\printbibliography}

\end{document}